\documentclass[11pt]{article}
\usepackage[margin=1in]{geometry}
\usepackage[latin2]{inputenc}
\usepackage[english]{babel}
\usepackage{amsmath}
\usepackage[all=normal,paragraphs=tight,bibliography=tight]{savetrees}
\usepackage{todonotes}
\usepackage{amsthm}
\usepackage{amssymb}
\usepackage{microtype}
\usepackage{xspace}
\usepackage{comment}
\usepackage{letltxmacro}
\usepackage{comment}
\usepackage{enumerate}
\usepackage{tikz}
\usepackage{url}
\usepackage{eufrak}

\usepackage{graphicx}
\usepackage{caption}
\usepackage{subcaption}

\newtheorem{theorem}{Theorem}[section]
\newtheorem{lemma}[theorem]{Lemma}
\newtheorem{claim}{Claim}

\theoremstyle{definition}

\newcommand{\Oh}{\ensuremath{\mathcal{O}}}

\newcommand{\clString}{{\sc{Closest String}}\xspace}
\newcommand{\clSubstring}{{\sc{Closest Substring}}\xspace}

\newcommand{\cPattern}{{\sc{Consensus Patterns}}\xspace}
\newcommand{\clique}{{\sc{Clique}}\xspace}

\newcommand{\Len}{L}
\newcommand{\dst}{d}
\newcommand{\Cnst}{\mathcal{S}}
\newcommand{\Ham}{\mathcal{H}}
\newcommand{\Hw}{\mathcal{H}}
\newcommand{\eps}{\varepsilon}
\newcommand{\bLarge}{\alpha}
\newcommand{\bSmall}{\beta}
\newcommand{\orth}{\rho}
\newcommand{\bl}{\ell}
\newcommand{\Sel}{\mathcal{T}}
\newcommand{\Forb}{\mathcal{F}}
\newcommand{\diam}[1]{\overline{#1}}

\def\cqedsymbol{\ifmmode$\lrcorner$\else{\unskip\nobreak\hfil
\penalty50\hskip1em\null\nobreak\hfil$\lrcorner$
\parfillskip=0pt\finalhyphendemerits=0\endgraf}\fi} 

\newcommand{\cqed}{\renewcommand{\qed}{\cqedsymbol}}

\newcommand{\executeiffilenewer}[3]{%
\ifnum\pdfstrcmp{\pdffilemoddate{#1}}%
{\pdffilemoddate{#2}}>0%
{\immediate\write18{#3}}\fi%
} 
\title{Lower bounds for approximation schemes for {\sc{Closest String}}}
\author{
  Marek Cygan\thanks{
    Institute of Informatics, University of Warsaw, Poland, \texttt{cygan@mimuw.edu.pl}.
    Supported by Polish National Science Centre grant DEC-2012/05/D/ST6/03214.
  }
  \and
  Daniel Lokshtanov\thanks{
    Department of Informatics, University of Bergen, Norway, \texttt{daniello@ii.uib.no}.
Supported by the BeHard grant under the
recruitment programme of the of Bergen Research Foundation.
  }
  \and
  Marcin Pilipczuk\thanks{
    Institute of Informatics, University of Warsaw, Poland, \texttt{marcin.pilipczuk@mimuw.edu.pl}.
    Supported by Polish National Science Centre grant DEC-2012/05/D/ST6/03214.
  }
  \and
  Micha\l{} Pilipczuk\thanks{
    Institute of Informatics, University of Warsaw, Poland, \texttt{michal.pilipczuk@mimuw.edu.pl}. Supported by Polish National Science Centre grant DEC-2013/11/D/ST6/03073 and by the Foundation for Polish Science via the START stipend programme. During the work on these results, Micha\l{} Pilipczuk held a post-doc position at Warsaw Center of Mathematics and Computer Science.}
  \and 
  Saket Saurabh\thanks{
    Institute of Mathematical Sciences, India, \texttt{saket@imsc.res.in}, and
    Department of Informatics, University of Bergen, Norway, \texttt{Saket.Saurabh@ii.uib.no}.
Supported by PARAPPROX, ERC starting grant no.~306992.
  }
}

\date{}

\begin{document}

\begin{titlepage}
\def\thepage{}
\thispagestyle{empty}
\maketitle

\begin{abstract}
In the \clString problem one is given a family $\Cnst$ of equal-length strings over some fixed alphabet, and the task is to find a string $y$ that minimizes the maximum Hamming distance between $y$ and a string from $\Cnst$. While polynomial-time approximation schemes (PTASes) for this problem are known for a long time [Li et al.; J. ACM'02], no {\em{efficient}} polynomial-time approximation scheme (EPTAS) has been proposed so far. In this paper, we prove that the existence of an EPTAS for \clString is in fact unlikely, as it would imply that $\mathrm{FPT}=\mathrm{W}[1]$, a highly unexpected collapse in the hierarchy of parameterized complexity classes. Our proof also shows that the existence of a PTAS for \clString with running time $f(\eps)\cdot n^{o(1/\eps)}$, for any computable function $f$, would contradict the Exponential Time Hypothesis.

\end{abstract}
\end{titlepage}

\section{Introduction}\label{sec:intro}
\clString and \clSubstring are two computational problems motivated by questions in molecular biology connected to identifying functionally similar regions of DNA or RNA sequences, as well as by applications in coding theory. In \clString we are given a family $\Cnst$ of strings over some fixed alphabet $\Sigma$, each of length $\Len$. The task is to find one string $y\in \Sigma^\Len$ for which $\max_{x\in \Cnst} \Ham(x,y)$ is minimum possible, where $\Ham(x,y)$ is the {\em{Hamming distance}} between $x$ and $y$, that is, the number of positions on which $x$ and $y$ have different letters. We will consider both the optimization variant of the problem where the said distance is to be minimized, and the decision variant where an upper bound $\dst$ is given on the input, and the algorithm needs to decide whether there exists a string $y$ with $\max_{x\in \Cnst} \Ham(x,y)\leq d$. \clSubstring is a more general problem where the strings from the input family $\Cnst$ all have length $m\geq \Len$, and we look for a string $y\in \Sigma^\Len$ that minimizes $\max_{x\in \Cnst}\min_{x'\textrm{ substring of }x} \Ham(x',y)$. In other words, we look for $y$ that can be fit as close as possible to a substring of length $\Len$ of each of the input strings from $\Cnst$.

Both \clString and \clSubstring, as well as numerous variations on these problems, have been studied extensively from the point of view of approximation algorithms. Most importantly for us, for both of these problems there are classic results providing {\em{polynomial-time approximation schemes}} ({\em{PTASes}}): for every $\eps>0$, it is possible to approximate in polynomial time the optimum distance within a multiplicative factor of $(1+\eps)$. The first PTASes for these problems were given by Li et al.~\cite{limawang}, and they had running time bounded by $n^{\Oh(1/\eps^4)}$. This was later improved by Andoni et al.~\cite{AndoniIP06} to $n^{\Oh(\frac{\log 1/\eps}{\eps^2})}$, and then by Ma and Sun~\cite{MaS09} to $n^{\Oh(1/\eps^2)}$, which constitutes the current frontier of  knowledge. We refer to the works~\cite{Boucher15,LiMW02a,GrammNR03,limawang,MaS09,Marx08} for a broad introduction to biological applications of \clString, \clSubstring, and related problems, as well as pointers to relevant literature.

One of the immediate questions stemming from the works of Li et al.~\cite{limawang}, Andoni et al.~\cite{AndoniIP06}, and Ma and Sun~\cite{MaS09}, is whether either for \clString or \clSubstring one can also give an {\em{efficient polynomial-time approximation scheme}} ({\em{EPTAS}}), i.e., an approximation scheme that for every $\eps>0$ gives a $(1+\eps)$-approximation algorithm with running time $f(\eps)\cdot n^{\Oh(1)}$, for some computable function $f$. In other words, the degree of the polynomial should be independent of $\eps$, whereas the exponential blow-up (inevitable due to NP-completeness) should happen only in the multiplicative constant standing in front of the running time. EPTASes are desirable from the point of view of applications, since they provide approximation algorithms that can be useful in practice already for relatively small values of $\eps$, whereas running times of general PTASes are usually prohibitive.

For the more general \clSubstring problem, this question was answered negatively by Marx~\cite{Marx08} using the techniques from parameterized complexity. More precisely, Marx considered various parameterizations of \clSubstring, and showed that when parameterized by $\dst$ and $|\Cnst|$, the problem remains W[1]-hard even for the binary alphabet. This means that the existence of a fixed-parameter algorithm with running time $f(\dst,|\Cnst|)\cdot n^{\Oh(1)}$, where $n$ is the total size of the input, would imply that $\mathrm{FPT}=\mathrm{W[1]}$, a highly unexpected collapse in the parameterized complexity. This result shows that, under $\mathrm{FPT}\neq \mathrm{W[1]}$, also an EPTAS for \clSubstring can be excluded. Indeed, if such an EPTAS existed, then by setting any $\eps<\frac{1}{\dst}$ one could in time $f(\dst)\cdot n^{\Oh(1)}$ distinguish instances with optimum distance value $\dst$ from the ones with optimum distance value $\dst+1$, thus solving the decision variant in fixed-parameter tractable (FPT) time. Using more precise results about the parameterized hardness of the \clique problem, Marx~\cite{Marx08} showed that, under the assumption of {\em{Exponential Time Hypothesis}} ({\em{ETH}}), which states that {\sc{3-SAT}} cannot be solved in time $\Oh(2^{\delta n})$ for some $\delta>0$, one even cannot expect PTASes for \clSubstring with running time $f(\eps)\cdot n^{o(\log (1/\eps))}$ for any computable function $f$. We refer to a survey of Marx~\cite{Marx08-survey} for more examples of links between parameterized complexity and the design of approximation schemes.

The methodology used by Marx~\cite{Marx08}, which is the classic connection between parameterized complexity and EPTASes that dates back to the work of Bazgan~\cite{bazgan:thesis} and of Cesati and Trevisan~\cite{CesatiT97}, completely breaks down when applied to \clString. This is because this problem actually does admit an FPT algorithm when parameterized by $\dst$. An algorithm with running time $d^d\cdot n^{\Oh(1)}$ was proposed by Gramm et al.~\cite{GrammNR03}. Later, Ma and Sun~\cite{MaS09} gave an algorithm with running time $2^{\Oh(d)}\cdot |\Sigma|^d\cdot n^{\Oh(1)}$, which is more efficient for constant-size alphabets. Both the algorithms of Gramm et al. and of Ma and Sun are known to be essentially optimal under ETH~\cite{LokshtanovMS11}, and nowadays they constitute textbook examples of advanced branching techniques in parameterized complexity~\cite{platypus}. Therefore, in order to settle the question about the existence of an EPTAS for \clString, one should look for a substantial refinement of the currently known techniques.

An approach for overcoming this issue was recently used by Boucher et al.~\cite{Boucher15}, who attribute the original idea to Marx~\cite{Marx08-survey}. Boucher et al. considered a problem called \cPattern, which is a variation of \clSubstring where the goal function is the total sum of Hamming distances between the center string and best-fitting substrings of the input strings, instead of the maximum among these distances. The problem admits a PTAS due to Li et al.~\cite{LiMW02a}, and was shown by Marx~\cite{Marx08} to be fixed-parameter tractable when parameterized by the target distance $\dst$. Despite the latter result, Boucher et al.~\cite{Boucher15} managed to prove that the existence of an EPTAS for \cPattern would imply that $\mathrm{FPT}=\mathrm{W[1]}$. The main idea is to provide a reduction from a W[1]-hard problem, such as \clique, where the output target distance $\dst$ is not bounded by a function of the input parameter $k$ (indeed, the existence of such a reduction would prove that $\mathrm{FPT}=\mathrm{W}[1]$), but the multiplicative gap between the optimum distances yielded for yes- and no-instances is $1+\frac{1}{g(k)}$, for some computable function $g$. Even though the output parameter is unbounded in terms of $k$, an EPTAS for the problem could be still used to distinguish between output instances obtained from yes- and no-instances of \clique in FPT time, thus proving that $\mathrm{FPT}=\mathrm{W[1]}$.

\paragraph*{Our contribution} In this paper we provide a negative answer to the question about the existence of an EPTAS for \clString by proving the following theorem.

\begin{theorem}\label{thm:closest-hardness}
The following assertions hold:
\begin{itemize}
\item Unless $\mathrm{FPT}=\mathrm{W[1]}$, there is no EPTAS for \clString over binary alphabet.
\item Unless ETH fails, there is no PTAS for \clString over binary alphabet with running time $f(\eps)\cdot n^{o(1/\eps)}$, for any computable function~$f$.
\end{itemize}
\end{theorem}

Thus, one should not expect an EPTAS for \clString, whereas for PTASes there is still a room for improvement between the running time of $n^{\Oh(1/\eps^2)}$ given by Ma and Sun~\cite{MaS09} and the lower bound of Theorem~\ref{thm:closest-hardness}. It is worth noting that our $f(\eps)\cdot n^{o(1/\eps)}$ time lower bound for $(1+\eps)$-approximating \clString also holds for the more general \clSubstring problem. This yields a significantly stronger lower bound than the previous $f(\eps)\cdot n^{o(\log (1/\eps))}$ lower bound of Marx~\cite{Marx08}.

Our proof of Theorem~\ref{thm:closest-hardness} follows the methodology proposed Marx~\cite{Marx08-survey} and used by Boucher et al.~\cite{Boucher15} for \cPattern. 
The following theorem, which is the main technical contribution of this work, states formally the properties of our reduction.

\begin{theorem}\label{thm:main}
There is an integer $c$ and an algorithm that, given an instance $(G,k)$ of \clique, works in time $2^k\cdot n^{\Oh(1)}$ and outputs an instance $(\Cnst,\Len,\dst)$ of \clString over alphabet $\{0,1\}$ with the following properties:
\begin{itemize}
\item If $G$ contains a clique on $k$ vertices, then there is a string $w\in \{0,1\}^\Len$ such that $\Ham(w,x)\leq \dst$ for each $x\in \Cnst$.
\item If $G$ does not contain a clique on $k$ vertices, then for each string $w\in \{0,1\}^\Len$ there is $x\in \Cnst$ such that $\Ham(w,x)>(1+\frac{1}{ck})\cdot \dst$. 
\end{itemize}
\end{theorem}

The statement of Theorem~\ref{thm:main} is similar to the core of the hardness proof of Boucher et al.~\cite{Boucher15}. However, our reduction is completely different from the reduction of Boucher et al., because the causes of the computational hardness of of \clString and \cPattern are quite orthogonal to each other. In \cPattern the difficulty lies in {\em picking} the right substrings of the input strings. Once these substrings are known the center string is easily computed in polynomial time, since we are minimizing the sum of the Hamming distances. In \clString there are no substrings to pick, we just have to find a center string for the given input strings. This is a computationally hard task because we are minimizing the maximum of the Hamming distances to the center, rather than the sum.

Theorem~\ref{thm:closest-hardness} follows immediately by combining Theorem~\ref{thm:main} with the known parameterized hardness results for \clique, gathered in the following theorem, and setting $\eps=\frac{1}{ck}$.

\begin{theorem}[cf. Theorem 13.25 and Corollary 14.23 of~\cite{platypus}]\label{thm:clique-hardness}
The following assertions hold:
\begin{itemize}
\item Unless $\mathrm{FPT}=\mathrm{W}[1]$, \clique cannot be solved in time $f(k)\cdot n^{\Oh(1)}$ for any computable function~$f$.
\item Unless ETH fails, \clique cannot be solved in time $f(k)\cdot n^{o(k)}$ for any computable function~$f$.
\end{itemize}
\end{theorem}

The main idea of the proof of Theorem~\ref{thm:main} is to encode the $n$ vertices of the given graph $G$ as an ``almost orthogonal'' family $\Sel$ of strings from $\{0,1\}^\bl$, for some $\bl=\Oh(\log n)$. Strings from $\Sel$ are used as identifiers of vertices of $G$, and the fact that they are almost orthogonal means that the identifiers of two distinct vertices of $G$ differ on approximately $\bl/2$ positions. On the other hand $\bl=\Oh(\log n)$, so the whole space of strings into which $V(G)$ is embedded has size polynomial in $n$. Using these properties, the reduction promised in Theorem~\ref{thm:main} is designed by a careful construction.

\paragraph*{Notation.}
By $\log p$ we denote the base-$2$ logarithm of $p$. For a positive integer $n$, we denote $[n]=\{1,2,\ldots,n\}$. The length of a string $x$ is denoted by $|x|$. For an alphabet $\Sigma$ and two equal-length strings $x,y$ over $\Sigma$, the {\em{Hamming distance}} between $x$ and $y$, denoted $\Ham(x,y)$, is the number of positions on which $x$ and $y$ have different letters. If $\Sigma=\{0,1\}$ is the binary alphabet, then the {\em{Hamming weight}} of a string $x$ over $\Sigma$, denoted $\Hw(x)$, is the number of $1$s in it. The {\em{complement}} of a string $x$ over a binary alphabet, denoted $\diam{x}$, is obtained from $x$ by replacing all $0$s with $1$s and vice versa. Note that if $|x|=|y|=n$, then $\Ham(x,y)=n-\Ham(\diam{x},y)=n-\Ham(x,\diam{y})=\Ham(\diam{x},\diam{y})$.

\section{Selection gadget}\label{sec:selection}
For the rest of this paper, we fix the following constants: $\orth=1/100$, $\bLarge=1/10$, $\bSmall=1/20$. 
Since $C$ is divisible by $100$, we have that that $\orth\bl$, $\bLarge\bl$, and $\bSmall\bl$ are all integers. 
First, we prove that among binary strings of logarithmic length one can find a linearly-sized family of ``almost orthogonal'' strings of balanced Hamming weight. The proof is by a simple greedy argument.

\begin{lemma}\label{lem:selection-strings}
There exist positive integers $C$ and $N$, where $C$ is divisible by $100$, with the following property. Let $n>N$ be any integer, and let us denote $\bl=C\cdot \lceil \log n\rceil$. Then there exists a set $\Sel\subseteq \{0,1\}^\bl$ with the following properties:
\begin{enumerate}
\item\label{pr:size} $|\Sel|=n$,
\item\label{pr:hamw} $\Hw(x)=\bl/2$ for each $x\in \Sel$, and
\item\label{pr:orth} $(1/2-\orth)\bl<\Ham(x,y)<(1/2+\orth)\bl$ for each distinct $x,y\in \Sel$.
\end{enumerate}
Moreover, given $n$, $\Sel$ can be constructed in time polynomial in $n$.
\end{lemma}
\begin{proof}
Let $H_2(\cdot)$ denote the binary entropy, i.e., $H_2(p)=-p\log p-(1-p)\log(1-p)$ for $p\in (0,1)$. Suppose $\bl$ is some positive integer divisible by $100$. Then it is well known that 
\begin{equation}\label{eq:entropy}
\sum_{i=0}^k\binom{\bl}{i}\leq 2^{\bl\cdot H_2(k/\bl)}
\end{equation}
for all integers $k$ with $0<k\leq \bl/2$; cf.~\cite[Lemma 16.19]{FlumGroheBook}. Let us denote
$$A=\sum_{i=0}^{(1/2-\orth)\bl} \binom{\bl}{i} + \sum_{i=(1/2+\orth)\bl}^\bl \binom{\bl}{i}.$$
Then from \eqref{eq:entropy} it follows that
$$A\leq 2\cdot 2^{\sigma\bl},$$
where $\sigma=H_2(1/2-\orth)<1$.

Suppose now that $\bl=C\cdot \lceil \log n\rceil$ for some positive integers $C$ and $n>1$, where $C$ is divisible by $100$. Then
\begin{eqnarray*}
n(\bl+1)\cdot A & \leq & 2n \cdot (C\lceil \log n\rceil+1)\cdot 2^{\sigma\cdot C\lceil \log n\rceil} \\
& \leq & 2\cdot (2C+1)\cdot 2^{\sigma C}\cdot n\log n \cdot 2^{\sigma\cdot C\log n}\\
& \leq & (4C+2)\cdot 2^{\sigma C}\cdot n^{\sigma C+2}.
\end{eqnarray*}
Since $\sigma<1$, we can choose $C$ to be an integer divisible by $100$ so that $\sigma C+2<C$. Then, we can choose $N$ large enough so that
$$(4C+2)\cdot 2^{\sigma C}\cdot n^{\sigma C+2} \leq n^C$$
for all integers $n>N$. Hence,
\begin{equation}\label{eq:greedy}
nA\leq \frac{n^C}{\bl+1}.
\end{equation}
We now verify that this choice of $C,N$ satisfies the required properties.

Consider the following greedy procedure performed on $\{0,1\}^\ell$. Start with $\Sel=\emptyset$ and all strings of $\{0,1\}^\bl$ marked as unused. In consecutive rounds perform the following:
\begin{enumerate}
\item Pick any $x\in \{0,1\}^\bl$ with $\Hw(x)=\bl/2$ that was not yet marked as used, and add $x$ to $\Sel$.
\item Mark every $y\in \{0,1\}^\bl$ with $\Ham(x,y)\leq(1/2-\orth)\ell$ or $\Ham(x,y)\geq(1/2+\orth)\bl$ as used.
\end{enumerate}
It is clear that at each step of the procedure, the constructed family $\Sel$ satisfies properties~\eqref{pr:hamw} and~\eqref{pr:orth}. Hence, it suffices to prove that the procedure can be performed for at least $n$ rounds. 

Note that the number of strings marked as used at each round is at most $A$. On the other hand, if $\mathcal{D}$ is the set of strings from $\{0,1\}^\bl$ that have Hamming weight exactly $\bl/2$, then 
$$|\mathcal{D}|\geq \frac{|\{0,1\}^\bl|}{\bl+1}=\frac{2^{C\lceil \log n\rceil}}{\bl+1}\geq \frac{n^C}{\bl+1}.$$
From \eqref{eq:greedy} we infer that $|\mathcal{D}|\geq nA$. This means that the algorithm will be able to find an unmarked $x\in \mathcal{D}$ for at least $n$ rounds, and hence to construct the family $\Sel$ with $|\Sel|=n$. It is easy to implement the algorithm in polynomial time using the fact that the size of $\{0,1\}^\bl$ is polynomial in $n$.  
\end{proof}

From now on, we adopt the constants $C,N$ given by Lemma~\ref{lem:selection-strings} to the notation. Let us also fix $n>N$; then let $\bl=C\cdot \lceil \log n\rceil$ and $\Sel$ be the set of strings given by Lemma~\ref{lem:selection-strings}, which we shall call {\em{selection strings}}. We define the set of {\em{forbidden strings}} $\Forb=\Forb(\Sel)$ as follows:
$$\Forb = \{ y \colon y\in \{0,1\}^\bl\textrm{ and $\Ham(x,y)\leq (1-\bLarge) \bl$ for all $x\in \Sel$}\}.$$
In other words, $\Forb$ comprises all the strings that are not almost diametrically opposite to some string from $\Sel$. The following lemma asserts the properties of $\Sel$ and $\Forb$ that we shall need later on.

\begin{lemma}\label{lem:forb}
Suppose $u\in \{0,1\}^\bl$. Then the following assertions hold:
\begin{enumerate}
\item\label{forb:in} If $u\in \Sel$, then $\Ham(u,y)\leq (1-\bLarge) \bl$ for each $y\in \Forb$.
\item\label{forb:far} If $\Ham(x,u)\geq \bSmall \bl$ for all $x\in \Sel$, then there exists $y\in \Forb$ such that $\Ham(u,y)\geq (1-\bSmall)\bl$.
\end{enumerate}
\end{lemma}
\begin{proof}
Property~\eqref{forb:in} follows directly from the definition of $\Forb$, so we proceed to the proof of~\eqref{forb:far}.

Suppose $\Ham(u,x)\geq \bSmall \bl$ for all $x\in \Sel$. If $\diam{u}\in \Forb$, then we could take $y=\diam{u}$, so suppose that $\diam{u}\notin \Forb$. This means that there exists $x_0\in \Sel$, for which $\Ham(x_0,\diam{u})>(1-\bLarge)\bl$; equivalently, $\Ham(\diam{x_0},\diam{u})<\bLarge \bl$. On the other hand, we have that $\Ham(x_0,u)\geq \bSmall \bl$, so also $\Ham(\diam{x_0},\diam{u})\geq\bSmall \bl$. Construct $y$ from $\diam{u}$ by taking any set of positions $X$ of size $\bSmall\bl$ on which $\diam{u}$ and $\diam{x_0}$ have the same letters, and flipping the letters on these positions (replacing $0$s with $1$s and vice versa). Such a set of positions always exists because $\bLarge+\bSmall<1$. Then we have that $\Ham(\diam{x_0},y)=\Ham(\diam{x_0},\diam{u})+\bSmall\bl$, which implies that
$$\bLarge \bl =\bSmall \bl+\bSmall\bl\leq \Ham(\diam{x_0},y)< (\bLarge+\bSmall)\bl.$$

We claim that $y\in \Forb$; suppose otherwise. Since $\Ham(\diam{x_0},y)\geq \bLarge \bl$, then also $\Ham(x_0,y)\leq (1-\bLarge)\bl$. As $y\notin \Forb$, there must exist some $x_1\in \Sel$, $x_0\neq x_1$, such that $\Ham(x_1,y)>(1-\bLarge)\bl$; equivalently $\Ham(\diam{x_1},y)<\bLarge\bl$. Hence, from the triangle inequality we infer that
$$\Ham(x_0,x_1)=\Ham(\diam{x_0},\diam{x_1})\leq \Ham(\diam{x_0},y)+\Ham(y,\diam{x_1})<(2\bLarge+\bSmall)\bl.$$
This is a contradiction with the assumption that $\Ham(x_0,x_1)\geq (1/2-\orth)\bl$, which is implied by $x_0,x_1\in \Sel$. Indeed, we have that $2\bLarge+\bSmall=\frac{1}{4}<\frac{49}{100}=1/2-\orth$.

Hence $y\in \Forb$. By definition we have that $\Ham(\diam{u},y)=\bSmall\bl$, which implies that $\Ham(u,y)=(1-\bSmall)\bl$. Thus, $y$ satisfies the required properties.
\end{proof}

\section{Main construction}\label{sec:construction}
\newcommand{\Cnstsel}{\Cnst_{\textrm{sel}}}
\newcommand{\Cnstadj}{\Cnst_{\textrm{adj}}}
\newcommand{\balOff}{\gamma}
\newcommand{\map}{\iota}

In this section we provide the proof of Theorem~\ref{thm:main}. Let $(G,k)$ be the input instance of \clique, and let $n=|V(G)|$. Let $C,N$ be the constants given by Lemma~\ref{lem:selection-strings}. We can assume that $n>N$, because otherwise the instance $(G,k)$ can be solved in constant time. Let $\bl=C\lceil \log n\rceil$. We run the polynomial-time algorithm given by Lemma~\ref{lem:selection-strings} that computes the set $\Sel\subseteq \{0,1\}^\bl$ of selection strings. Let $\Forb=\Forb(\Sel)$ be the set of forbidden strings, as defined in Section~\ref{sec:selection}. Note that $\Forb$ can be computed in polynomial time directly from the definition, due to $|\{0,1\}^\bl|=n^{\Oh(1)}$.

We now present the construction of the output instance $(\Cnst,\Len,\dst)$ of \clString. Set $\Len=k\bl+\balOff \bl$, where $\balOff=\orth+\bLarge=\frac{11}{100}$, and partition the set $[\Len]$ of positions in strings of length $\Len$ into $k+1$ {\em{blocks}}:
\begin{itemize}
\item $k$ blocks $B_i$ for $i\in [k]$ of length $\bl$ each, where $B_i=\{(i-1)\bl+1,(i-1)\bl+2,\ldots,i\bl\}$;
\item special {\em{balancing block}} $C$ of length $\balOff \bl$, where $C=\{k\bl+1,k\bl+2,\ldots,\Len\}$.
\end{itemize}
For $w\in \{0,1\}^\Len$ and a contiguous subset of positions $X$, by $w[X]$ we denote the substring of $w$ formed by positions from $X$.

Let us first discuss the intuition. The choice the solution string makes on consecutive blocks $B_i$ will encode a selection of a $k$-tuple of vertices in $G$. Vertices of $G$ will be mapped one-to-one to strings from $\Sel$. The family of constraint strings $\Cnst$ will consist of two subfamilies $\Cnstsel$ and $\Cnstadj$ with the following roles:
\begin{itemize}
\item Strings from $\Cnstsel$ ensure that on each block $B_i$, the solution picks a substring that is close to some element of $\Sel$. The selection of this element encodes the choice of the $i$th vertex from the $k$-tuple.
\item Strings from $\Cnstadj$ verify that vertices of the chosen $k$-tuple are pairwise different and adjacent, and hence they form a clique.
\end{itemize}
A small technical caveat is that for strings from $\Cnstsel$ and from $\Cnstadj$, the intended Hamming distance from the solution string will be slightly different. The role of the balancing block $C$ is to equalize this distance by a simple additional construction.

We proceed to the formal description. Since $|V(G)|=|\Sel|$, let $\map\colon V(G)\to \Sel$ be an arbitrary bijection. 

The family $\Cnstsel$ consists of strings $a(i,y,\phi,z)$, for all $i\in [k]$, $y\in \Forb$, $\phi$ being a function from $[k]\setminus \{i\}$ to $\{0,1\}$, and $z$ being a binary string of length $\balOff \bl$. String $a(i,y,\phi,z)$ is constructed as follows:
\begin{itemize}
\item On block $B_i$ put the string $y$.
\item For each $j\in [k]\setminus \{i\}$, on block $B_j$ put a string consisting of $\bl$ zeroes if $\phi(j)=0$, and a string consisting of $\bl$ ones if $\phi(j)=1$.
\item On balancing block $C$ put the string $z$.  
\end{itemize}
Thus, $|\Cnstsel|=k\cdot |\Forb|\cdot 2^{k-1}\cdot 2^{\balOff \bl}\leq 2^k\cdot n^{\Oh(1)}$. Also, $\Cnstsel$ can be constructed in time $2^k\cdot n^{\Oh(1)}$ directly from the definition.

The family $\Cnstadj$ consists of strings $b(i,j,(u,v),\psi)$, for all $i,j\in [k]$ with $i<j$, $(u,v)$ being an ordered pair of vertices of $G$ that are either equal or non-adjacent, and $\psi$ being a function from $[k]\setminus \{i,j\}$ to $\{0,1\}$. String $b(i,j,(u,v),\psi)$ is constructed as follows:
\begin{itemize}
\item On block $B_i$ put the string $\diam{\map(u)}$.
\item On block $B_j$ put the string $\diam{\map(v)}$.
\item On block $B_q$, for $q\in [k]\setminus \{i,j\}$, put a string consisting of $\bl$ zeroes if $\psi(q)=0$, and a string consisting of $\bl$ ones if $\psi(q)=1$.
\item On balancing block $C$ put a string consisting of $\balOff\bl$ zeroes.
\end{itemize}
Thus, $|\Cnstadj|\leq \binom{k}{2}\cdot n^2\cdot 2^{k-2}\leq 2^k\cdot n^{\Oh(1)}$. Again, $\Cnstadj$ can be constructed in time $2^k\cdot n^{\Oh(1)}$ directly from the definition.

Set $\Cnst=\Cnstsel\cup \Cnstadj$ and $\dst=(k/2+1/2+\orth)\cdot \bl$. This concludes the construction. Its correctness will be verified in two lemmas that mirror the properties listed in Theorem~\ref{thm:main}.

\begin{lemma}\label{lem:completeness}
If $G$ contains a clique on $k$ vertices, then there exists a string $w\in \{0,1\}^\Len$ such that $\Ham(w,x)\leq d$ for each $x\in \Cnst$.
\end{lemma}
\begin{proof}
Let $\{c_1,c_2,\ldots,c_k\}$ be a $k$-clique in $G$. Construct $w$ by putting $\map(c_i)$ on block $B_i$, for each $i\in [k]$, and zeroes on all the positions of the balancing block $C$.

First, take any string $a=a(i,y,\phi,z)\in \Cnstsel$. Since $\map(c_i)\in \Sel$ and $y\in \Forb$, by Lemma~\ref{lem:forb}\eqref{forb:in} we infer that $\Ham(w[B_i],a[B_i])=\Ham(\map(c_i),y)\leq (1-\bLarge) \bl$. For each $j\in [k]\setminus \{i\}$, since $\Hw(\map(c_j))=\bl/2$ due to $\map(c_j)\in \Sel$, we have that $\Ham(w[B_j],a[B_j])=\Ham(\map(c_j),a[B_j])=\bl/2$, regardless of the value of $\phi(j)$. Finally, obviously $\Ham(w[C],a[C])\leq |C|=\balOff\bl$. Hence 
$$\Ham(w,a)\leq (1-\bLarge) \bl+(k-1)\bl/2+\balOff\bl=d.$$

Second, take any string $b=b(i,j,(u,v),\psi)\in \Cnstadj$. Since $c_i$ and $c_j$ are different and adjacent, whereas $u$ and $v$ are either equal or non-adjacent, we have $(c_i,c_j)\neq (u,v)$. Without loss of generality suppose that $c_i\neq u$; the second case will be symmetric. Then $\Ham(w[B_i],b[B_i])=\Ham(\map(c_i),\diam{\map(u)})\leq (1/2+\orth)\bl$, due to property \eqref{pr:orth} of Lemma~\ref{lem:selection-strings}. Obviously, $\Ham(w[B_j],b[B_j])\leq |B_j|\leq \bl$. Finally, for every $q\in [k]\setminus \{i,j\}$ we have that $\Hw(\map(c_q))=\bl/2$, and hence $\Ham(w[B_q],b[B_q])=\Ham(\map(c_q),b[B_q])=\bl/2$, regardless of the value of $\psi(q)$. Strings $w$ and $b$ match on positions of $C$, so $\Ham(w[C],b[C])=0$. Summarizing,
$$\Ham(w,b)\leq (1/2+\orth)\bl+\bl+(k-2)\bl/2=d.$$\end{proof}

\begin{lemma}\label{lem:soundness}
If there is a string $w\in \{0,1\}^\Len$ such that $\Ham(w,x)<d+\bSmall \bl$ for each $x\in \Cnst$, then $G$ contains a clique on $k$ vertices.
\end{lemma}
\begin{proof}
We first prove that on each block $B_i$, $w$ is close to selecting an element of $\Sel$.
\begin{claim}\label{cl:vertices}
For each $i\in [k]$ there exists a unique $x_i\in \Sel$ such that $\Ham(w[B_i],x_i)<\bSmall \bl$.
\end{claim}
\begin{proof}
Uniqueness follows directly from property \eqref{pr:orth} of Lemma~\ref{lem:selection-strings} and the triangle inequality, so it suffices to prove existence.

Let $u=w[B_i]$. For the sake of contradiction, suppose $\Ham(u,x)\geq \bSmall \bl$ for each $x\in \Sel$. From Lemma~\ref{lem:forb}\eqref{forb:far} we infer that there exists $y\in \Forb$ such that $\Ham(u,y)\geq (1-\bSmall)\bl$. Let us take $\phi\colon [k]\setminus \{i\}\to \{0,1\}$ defined as follows: $\phi(j)=0$ if in $w$ the majority of positions of $B_j$ contain a one, and $\phi(j)=1$ otherwise. Also, define $z=\diam{w[C]}$. Consider string $a=a(i,y,\phi,z)\in \Cnstsel$. Then, it follows that
\begin{itemize}
\item $\Ham(w[B_i],a[B_i])=\Ham(u,y)\geq (1-\bSmall)\bl$;
\item $\Ham(w[B_j],a[B_j])\geq \bl/2$ for each $j\in [k]\setminus \{i\}$;
\item $\Ham(w[C],a[C])=\Ham(w[C],\diam{w[C]})=|C|=\balOff \bl$.
\end{itemize}
Consequently,
$$\Ham(w,a)\geq (1-\bSmall)\bl+(k-1)\bl/2+\balOff\bl=d+\bSmall\bl.$$
This is a contradiction with the assumption that $\Ham(w,x)<d+\bSmall \bl$ for each $x\in \Cnst$.
\cqed\end{proof}

For each $i\in [k]$, let $c_i=\map^{-1}(x_i)$.

\begin{claim}\label{cl:edges}
For all $i,j\in [k]$ with $i<j$, vertices $c_i$ and $c_j$ are different and adjacent.
\end{claim}
\begin{proof}
For the sake of contradiction, suppose $c_i$ and $c_j$ are either equal or non-adjacent. Define $\psi\colon [k]\setminus \{i,j\}\to \{0,1\}$ as follows: $\psi(q)=0$ if in $w$ the majority of positions of $B_q$ contain a one, and $\psi(q)=1$ otherwise. Then, for $(c_i,c_j)$ we have constructed string $b=b(i,j,(c_i,c_j),\psi)\in \Cnstadj$. Observe now that
\begin{itemize}
\item $\Ham(w[B_i],b[B_i])=\Ham(w[B_i],\diam{x_i})>(1-\bSmall)\bl$, since $\Ham(w[B_i],x_i)<\bSmall\bl$;
\item Similarly, $\Ham(w[B_j],b[B_j])>(1-\bSmall)\bl$;
\item $\Ham(w[B_q],b[B_q])\geq \bl/2$ for each $q\in [k]\setminus \{i,j\}$;
\item $\Ham(w[C],b[C])\geq 0$.
\end{itemize}
Consequently,
$$\Ham(w,b)\geq 2(1-\bSmall)\bl+(k-2)\bl/2=(k/2+1-2\bSmall)\bl>d+\bSmall\bl.$$
This is a contradiction with the assumption that $\Ham(w,x)<d+\bSmall \bl$ for each $x\in \Cnst$.
\cqed\end{proof}

Claim~\ref{cl:edges} asserts that, indeed, $\{c_1,c_2,\ldots,c_k\}$ is a $k$-clique in $G$.
\end{proof}

Lemmas~\ref{lem:completeness} and~\ref{lem:soundness} conclude the proof of Theorem~\ref{thm:main}, where $c$ can be taken to be any constant larger than $\frac{\dst}{\bSmall\ell\cdot k}\leq \frac{2}{\bSmall}=40$.

\section{Conclusions}\label{sec:conclusions}
In this paper we have proved that \clString does not have an EPTAS under the assumption of $\mathrm{FPT}\neq\mathrm{W[1]}$. Moreover, under the stronger assumption of the Exponential Time Hypothesis, one can also exclude PTASes with running time $f(\eps)\cdot n^{o(1/\eps)}$, for any computable function $f$. However, the fastest currently known approximation scheme for \clString has running time $n^{\Oh(1/\eps^2)}$~\cite{MaS09}. This leaves a significant gap between the known upper and lower bounds. Despite efforts, we were unable to close this gap, and hence we leave it as an open problem.

\bibliographystyle{abbrv}
\bibliography{closest-string}

\end{document}